\documentclass[12pt]{article}%
\usepackage{amsfonts}
\usepackage{amsmath,amssymb,amsthm,enumerate,epsfig,graphicx}
\usepackage{ifthen,latexsym,syntonly,natbib}
\usepackage{amsmath}
\usepackage{amssymb}
\usepackage{graphicx}%
\nonstopmode
\newtheorem{theorem}{Theorem}

\newtheorem{lemma}{Lemma}

\newtheorem{remark}{Remark}

\makeatletter
\@addtoreset{equation}{section}
\renewcommand\@biblabel[1]{}
\def\baselinestretch{1.0}
\begin{document}

\title{Portfolio Choice in Markets with Contagion}
\date{
%TCIMACRO{\TeXButton{today's date}{\today}}%
%BeginExpansion
\today
%EndExpansion
}
\author{Yacine A\"{\i}t-Sahalia\thanks{Research supported by the NSF\ under grant
NES-0850533.}\\{\small Department of Economics}\\{\small Princeton University }\\{\small and NBER}\bigskip
\and T. R. Hurd\thanks{Research supported by NSERC and MITACS Canada.}\\{\small Dept. of Mathematics }\\{\small and Statistics}\\{\small McMaster University}}
\maketitle

\begin{abstract}
\noindent{ We consider the problem of optimal investment and consumption in a
class of multidimensional jump-diffusion models in which asset prices are
subject to mutually exciting jump processes. This captures a type of contagion
where each downward jump in an asset's price results in increased likelihood
of further jumps, both in that asset and in the other assets. We solve in
closed-form the dynamic consumption-investment problem of a log-utility
investor in such a contagion model, prove a theorem verifying its optimality
and discuss features of the solution, including flight-to-quality. The
exponential and power utility investors are also considered: in these cases,
the optimal strategy can be characterized as a distortion of the strategy of a
corresponding non-contagion investor. }



\end{abstract}

\thispagestyle{empty} \ \newpage\setcounter{page}{1}
%TCIMACRO{\TeXButton{1.5-spaced document, normal fontsize}{\renewcommand
%{\baselinestretch}{1.5}
%\normalsize} }%
%BeginExpansion
\renewcommand{\baselinestretch}{1.5}
\normalsize
%EndExpansion
\setcounter{equation}{0}

\section{Introduction}

The recent financial crisis has emphasized the relevance of jumps for
understanding the various forms of risk inherent in asset returns and their
implications for asset allocation and diversification. Most portfolios, from
those of individual investors to those of more sophisticated institutional
investors including University endowments, suffered badly during the latest
crisis episode, with many commonly employed asset allocation strategies
resulting in large losses in 2008-09. Reasonably diversified portfolios can
survive a single isolated negative jump in asset returns. However, jumps that
tend to affect most or all asset classes together are difficult to hedge by
diversification alone. Moreover, additional jumps of this nature seemed to
happen in close succession, as if the very occurrence of a jump substantially
increased the likelihood of future jumps.

Motivated by these events, we consider in this paper the issue of optimal
portfolio construction when assets are subject to jumps that share the
qualitative features experienced most vividly during the recent financial
crisis. These salient features include the fact that multiple jumps were
observed, at a rate that was markedly higher than the long term unconditional
arrival rate; these jumps affected multiple asset classes and markets; and
they affected them not necessarily at the same time, but typically in close
succession over days or weeks.

To capture these key elements, we consider a model for asset returns where a
jump in one asset class or region raises the probability of future jumps in
both the same asset class or region, and the other classes or regions. Jump
processes of this type were first introduced by \cite{hawkes71a} with further
developments due to \cite{HawkesOakes74} and \cite{Oakes75}. Models of this
type have been used in epidemiology, neurophysiology and seismology (see,
e.g., \cite{brillinger88} and \cite{OgataAkaike82}), genome analysis (see
\cite{reynaudbouretschbath10}), credit derivatives
(\cite{erraisgieseckegoldberg10}), to model transaction times and price
changes at high frequency (\cite{bowsher07}), trading activity at different
maturities of the yield curve (\cite{salmontham}) and propagation phenomena in
social interactions (\cite{cranesornette08}.)

We extend the pure jump Hawkes model employed in the above applications, in
order to better represent financial asset returns. We add a drift to capture
the assets' expected returns and a standard, Brownian-driven, volatility
component to capture their day-to-day normal variations. We call this model a
\textquotedblleft Hawkes jump-diffusion\textquotedblright\ by analogy with the
Poisson jump-diffusion of \cite{merton76}. Unlike models typically employed in
finance, the jump part of this model is no longer a L\'{e}vy process since
excitation introduces a departure from independence of the increments of the jumps.

In the model, jump intensities are stochastic and react to recent jumps: a
jump increases the rate of incidence (or intensity) of future jumps; running
counter to this is mean reversion, which pulls the jump intensities back down
in the absence of further excitation. In the univariate case, only
\textquotedblleft self excitation\textquotedblright\ can take place, whereas
in the multivariate case \textquotedblleft mutual excitation\textquotedblright%
\ consisting of both self- and cross-excitation (from one asset to
another)\ can take place. We now illustrate the presence of the mutual
excitation phenomenon by filtering jump intensities for jumps in the US
financial sector stock index during the recent crisis. Figure
\ref{fig:ach_datafigure} plots the estimated intensity of US\ jumps over time,
filtered from the observed returns on indices of financial stocks in the US,
UK, Eurozone and Asia. The excitation mechanism is apparent in the Fall of
2008, when jump intensities increase rapidly in response to each jump, most of
them originating in the US, and to a lesser extent in the Winter of 2009,
which looks more like a slow train wreck.

The bottom panel of the plot shows that the filtered intensities contain
information that is different from other measures of market stress, VIX and
the CDS rate on financial stocks. In particular, VIX\ is a measure of total
quadratic variation and as a result captures the total risk of the assets
instead of just their jump risk. So the same jumps which cause the jump
intensity to increase in the middle panel also cause VIX\ to increase in the
bottom panel, but VIX includes Brownian volatility, making it\ a much noisier
measure of jump risk.%

%TCIMACRO{\FRAME{ftbpFU}{7.0543in}{5.0963in}{0pt}{\Qcb{Time-varying Jump
%Intensities During the Financial Crisis. \ \ \ The top panel shows time series
%of stock indices for the financial sector in the four regions, 2007-2009. The
%middle panel shows the intensity of US\ jumps filtered from the model, based
%on each 3$\sigma$ and above event identified as a jump. Each jump leads to an
%increase in the jump intensity, followed by mean reversion until the next
%jump. The bottom panel shows the time series of two alternative measures of
%financial distress, the VIX\ index and Markit's CDS\ index for the banking
%sector in the US.}}{\Qlb{fig:ach_datafigure}}{ach_datafigure_3panel.eps}%
%width 7.0543in;  height 5.0963in;  depth 0pt;  original-width 13.3363in;
%original-height 9.6193in;  cropleft "0";  croptop "1";  cropright "1";
%filename 'ACH_datafigure_3panel.eps';file-properties "XNPEU";}}}%
%BeginExpansion
\begin{figure}
[ptb]
\begin{center}
\includegraphics[
height=5.0963in,
width=7.0543in
]%
{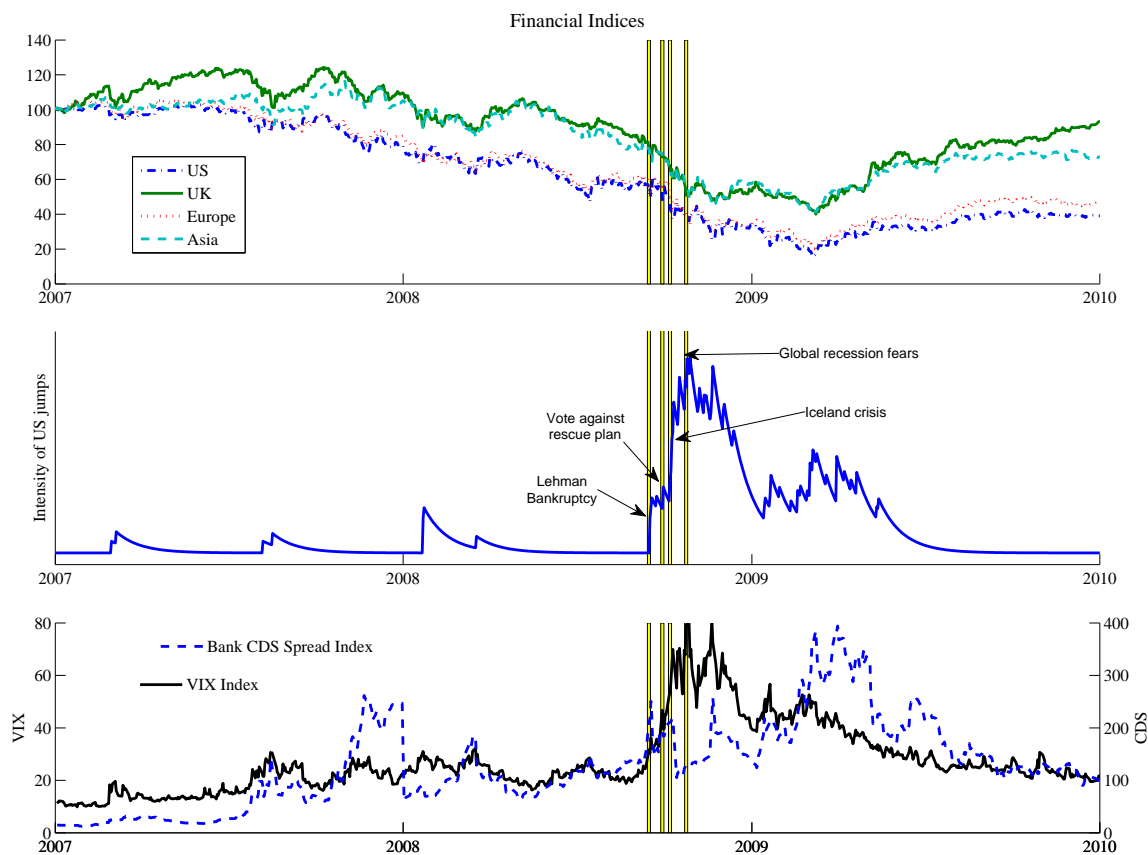}%
\caption{Time-varying Jump Intensities During the Financial Crisis. \ \ \ The
top panel shows time series of stock indices for the financial sector in the
four regions, 2007-2009. The middle panel shows the intensity of US\ jumps
filtered from the model, based on each 3$\sigma$ and above event identified as
a jump. Each jump leads to an increase in the jump intensity, followed by mean
reversion until the next jump. The bottom panel shows the time series of two
alternative measures of financial distress, the VIX\ index and Markit's
CDS\ index for the banking sector in the US.}%
\label{fig:ach_datafigure}%
\end{center}
\end{figure}
%EndExpansion

The purpose of the paper is to solve for the optimal portfolio of an investor
who faces this type of risk in his/her investment opportunity set. By
considering a more realistic model for jumps, incorporating mutual excitation,
we are able to study the optimal portfolio of an investor in a realistic
setting where a jump that occurs somewhere will increase the probability of
further jumps in other asset classes or markets. The model generates jumps
that will tend to be clustered (as a result of the time series
self-excitation), systematic (as a result of the cross-sectional
excitation),\ but neither exactly simultaneous nor certain, since the
excitation phenomenon merely raises the probability of future jump occurrence.
By analogy with epidemics, the probability of getting infected increases in a
pandemy but does not typically reach one, and there is an incubation period
which can range in the case of financial markets from hours to days, depending
upon subsequent news arrival, and once established the pandemy does not go
away immediately. Furthermore, the model is multivariate and the contagion can
be asymmetric, with jumps occurring in one asset class or market having a
greater excitation potential for the other sectors or regions than jumps that
originate elsewhere: for instance, most financial crises that originate in or
transit through the US tend to have greater ramifications in the rest of the
world than crises that originate outside the US. Poisson jumps, whose
intensities are constant, are not able to reproduce these empirical features,
and this motivates our inclusion of the more general class of Hawkes jumps in
the model.

This paper is part of a literature that has investigated the properties of
optimal portfolios when asset returns can jump (see, e.g., \cite{aase84},
\cite{jeanblancpontier90}, \cite{shirakawa90}, \cite{hanrachev00},
\cite{ortobellihuberrachevschwartz03}, \cite{kallsen00}, \cite{CarrJinMadan01}%
, \cite{liulongstaffpan03}, \cite{dasuppal04}, \cite{emmerkluppelberg04},
\cite{madan04}, \cite{CvitanicPolimenisZapatero08}, \cite{delongkluppelberg08}
and \cite{ach09}). The novel aspect in the present paper is the inclusion of
Hawkes jumps in asset returns: such jumps share the dual characteristics of
being systematic, meaning that they affect multiple assets or asset classes at
the same time, and mutually exciting, meaning that they affect the rate at
which future jumps occur in each asset class. By contrast, in the earlier
model of \cite{ach09}, assets were subject to random jumps which could affect
one or more asset or asset classes, but when they occurred, they were
simultaneous and every asset in that sector or region would jump. Such jumps
were also \textquotedblleft Poissonian\textquotedblright, in the sense that
the arrival of jumps today did not influence the future arrival of jumps.

We show that the optimal portfolio solution in the model can be obtained in
full closed-form in the log utility case, and in quasi-closed-form in other
cases. Importantly, we show that the optimal solution becomes time-varying,
with the investor reacting to changes in the intensity of the jumps. For a
log-utility investor, the solution remains myopic, as in the classical
\cite{merton71} problem with log-utility, in the sense that the investor does
not need to take into account the full dynamics of the state variables. The
log investor in our model holds at each point in time the same portfolio as a
log investor who believes that jump intensities are constant, but his/her
optimal portfolio weight is now constantly changing to reflect the
time-variation in jump intensities. One consequence of this result is that
each time a market shock occurs, the investor perceives an increase in jump
intensities, and sells some amount of \textit{each} risky asset, and invests
the proceeds in the riskless asset, a behavior we interpret as a
\textquotedblleft flight to quality.\textquotedblright

Formal computations also work for both power and exponential utility
investors, although we have not proved the appropriate verification theorem
for these utilities. Nevertheless, the resultant investment strategy under
contagion can be interpreted in terms of the equivalent strategy under the
\textquotedblleft non-contagion\textquotedblright\ assumption that jump
intensities are constant. We find that under the contagion conditions of the
model, the investor will choose a portfolio that is optimal for a
non-contagion investor who has a specific distorted value of the intensities,
which we characterize. This distortion of intensities has the effect of
magnifying the investment in the risky assets: the contagion investor will go
\textquotedblleft longer\textquotedblright\ when the non-contagion investor is
long in the risky asset, and will go \textquotedblleft
shorter\textquotedblright\ when the non-contagion investor is short in the
risky asset.

The paper is organized as follows. Section \ref{sec:model} presents the model
for asset returns. Section \ref{sec:choice} introduces the optimal portfolio
problem when jumps are mutually exciting in the general case. Section
\ref{sec:logU} specializes the solution to the case of an investor with
log-utility and derives the optimal portfolio and consumption policy in
closed-form, including a complete verification theorem that supports this
policy. Section \ref{sec:sector} develops some interesting market
specifications which exhibit an explicit closed-form log-optimal policy.
Section \ref{sec:conseq} explores some of the properties of such explicit
asset allocation policies. The exponential and power utility investment
problems are outlined in Section \ref{sec:powerandexp}. In these problems, the
solutions can be characterized as distorted versions of the non-contagion
solutions. Section \ref{sec:conclusions} concludes.

\section{Mutually Exciting Jumps\label{sec:model}}

In this paper, jumps in one asset class not only increase the probability of
future jumps in that asset class (self excitation)\ but also in other asset
classes (cross excitation). In a mutually exciting model, the intensity of a
jump counting process $N$ ramps up in response to past jumps. A mutually
exciting process is a special case of path-dependent point process, whose
intensity depends on the path of the underlying process. Mutually exciting
counting processes, $N_{l,t}$ ($l=1,...,m$), form an $m$-vector
$\boldsymbol{N}_{t}=\left[  N_{1,t},...,N_{m,t}\right]  ^{\prime}$ such
that\footnote{In this paper, we work in a filtered probability space
$(\Omega,\mathcal{F},(\mathcal{F}_{t})_{t\geq0},\mathbb{P})$ that satisfies
\textquotedblleft the usual conditions.\textquotedblright}%
\begin{align}
\mathbb{P}\left[  N_{l,t+\Delta t}-N_{l,t}=1|\mathcal{F}_{t}\right]   &
=\lambda_{l,t}\Delta t+o\left(  \Delta t\right)  ,\label{eq:Hawkes_pr}\\
\mathbb{P}\left[  N_{l,t+\Delta t}-N_{l,t}>0|\mathcal{F}_{t}\right]   &
=o\left(  \Delta t\right) \nonumber
\end{align}
independently for each $l$. In the standard model specification we adopt in
this paper, the Hawkes intensity processes $\boldsymbol{\lambda}_{t}=\left[
\lambda_{1,t},...,\lambda_{m,t}\right]  ^{\prime}$ have the integrated form%
\begin{equation}
\lambda_{l,t}=e^{-\alpha_{l}t}\lambda_{l,0}+(1-e^{-\alpha_{l}t})\lambda
_{l,\infty}+\sum\nolimits_{j=1}^{m}\int_{0}^{t}d_{lj}e^{-\alpha_{l}%
(t-s)}dN_{j,s},\text{ }l=1,...,m \label{eq:Hawkes_intensity}%
\end{equation}
For all $l,j=1,\ldots,m$ the parameters $\lambda_{l,\infty},d_{lj}\geq0$ and
$\lambda_{l,0},\alpha_{l}>0$ are constants. These parameter restrictions
ensure the positivity of the intensity processes with probability one.

Differentiation of equation (\ref{eq:Hawkes_intensity}) shows that the
intensity in asset class $l$ has dynamics given by%
\begin{equation}
d\lambda_{l,t}=\alpha_{l}\left(  \lambda_{l,\infty}-\lambda_{l,t}\right)
dt+\sum\nolimits_{j=1}^{m}d_{lj}dN_{j,t}. \label{lambdadynamics}%
\end{equation}
In other words, a jump $dN_{j,s},$ occurring at time $s\in\lbrack0,t)$ in
asset class $j=1,\ldots,m,$ raises each of the jump intensities $\lambda
_{l,t}$, $l=1,\ldots,m,$ by a constant amount $d_{lj}$. The $l$th jump
intensity then mean-reverts to level $\lambda_{l,\infty}$ at speed $\alpha
_{l}$ until the next jump occurs. Equation \ref{lambdadynamics} also reveals
that $(\boldsymbol{N},\boldsymbol{\lambda})$ is a $2m$-dimensional Markov
process, while $(\boldsymbol{\lambda})$ alone is an $m$-dimensional Markov process.

As a result, our model generates clusters of jumps over time, and jumps can
propagate at different speed and with different intensities in the different
asset classes depending on where they originate and which path they take to
reach a given asset class or market. Free parameters in the model control the
extent to which the two forms of excitation take place, the relative strength
of the contagion phenomenon in different directions and the speed with which
the excitation takes place and then relaxes.

The model produces both cross-asset class and time series excitation. In the
univariate self-exciting case, a typical sample path of one component of
$\boldsymbol{\lambda}$ is illustrated in Figure \ref{fig:curveint}. Each jump
increases the jump intensity, followed by mean-reversion.%

%TCIMACRO{\FRAME{ftbFU}{5.6015in}{3.0237in}{0pt}{\Qcb{Sample path of a Hawkes
%intensity, $\lambda_{l,t}.$}}{\Qlb{fig:curveint}}{ACH_figcurveint.eps}%
%width 5.6015in;  height 3.0237in;  depth 0pt;  original-width 5.5443in;
%original-height 2.9801in;  cropleft "0";  croptop "1";  cropright "1";
%cropbottom "0";  filename 'ACH_figcurveint.eps';file-properties "XNPEU";}}}%
%BeginExpansion
\begin{figure}
[tb]
\begin{center}
\includegraphics[
height=3.0237in,
width=5.6015in
]%
{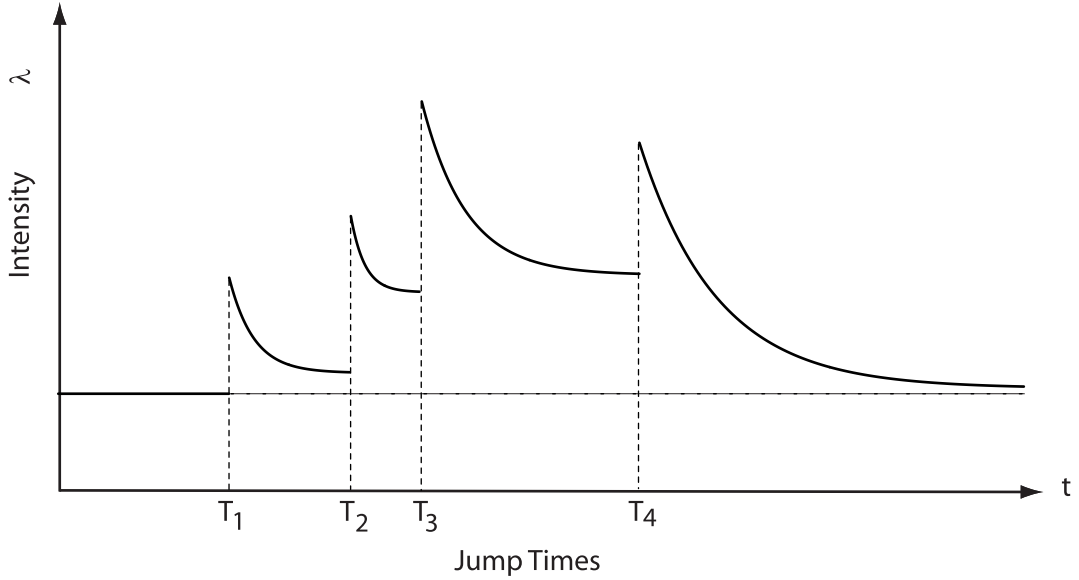}%
\caption{Sample path of a Hawkes intensity, $\lambda_{l,t}.$}%
\label{fig:curveint}%
\end{center}
\end{figure}
%EndExpansion
This model also produces cross-asset, or mutual, excitation. Jumps in asset
class $l$ that occurred $u$ units of time into the past raise the intensity of
jumps in asset class $j$ by $d_{jl}e^{-\alpha_{j}u},$ while conversely jumps
in asset class $j$ raise the intensity of jumps in asset class $l$ by
$d_{lj}e^{-\alpha_{l}u}.$ Jumps in a given asset class $i$ also raise the
intensity of future jumps in the same asset class. Figure
\ref{fig:2assetNSlambda} illustrates this with two assets. At time $T_{1}$
there is a jump in the first asset class value, $S_{1}$. This jump
self-excites the jump intensity $\lambda_{1}$. This increase in $\lambda_{1}$
raises the probability of observing another jump in $S_{1}$ at the future time
$T_{2}$. These jumps have a contagious effect on $S_{2}$ since a jump in
$S_{1}$ cross-excites the jump intensity of $S_{2}$. This, in turn, raises the
probability of seeing a jump in $S_{2}$ at time $T_{3}$. Latter on, at time
$T_{4}$, the jump in $S_{2}$ raises the probability of seeing a jump in
$S_{1}$ at some future time $T_{5}$. The degree to which self- and
cross-excitation matter in the model, and their relative strengths, is
controlled by the parameters in \textbf{$d$} and $\boldsymbol{\alpha}.$%

%TCIMACRO{\FRAME{ftphFU}{5.1224in}{5.3366in}{0pt}{\Qcb{Cross- and
%self-excitation in a two asset-class world:\ Sample paths of the jumps, asset
%prices and jump intensities.}}{\Qlb{fig:2assetNSlambda}}{picturecontagion.eps}%
%width 5.1224in;  height 5.3366in;  depth 0pt;  original-width 7.2774in;
%original-height 7.5844in;  cropleft "0";  croptop "1";  cropright "1";
%cropbottom "0";  filename 'picturecontagion.eps';file-properties "NPEU";}}}%
%BeginExpansion
\begin{figure}
[pth]
\begin{center}
\includegraphics[
height=5.3366in,
width=5.1224in
]%
{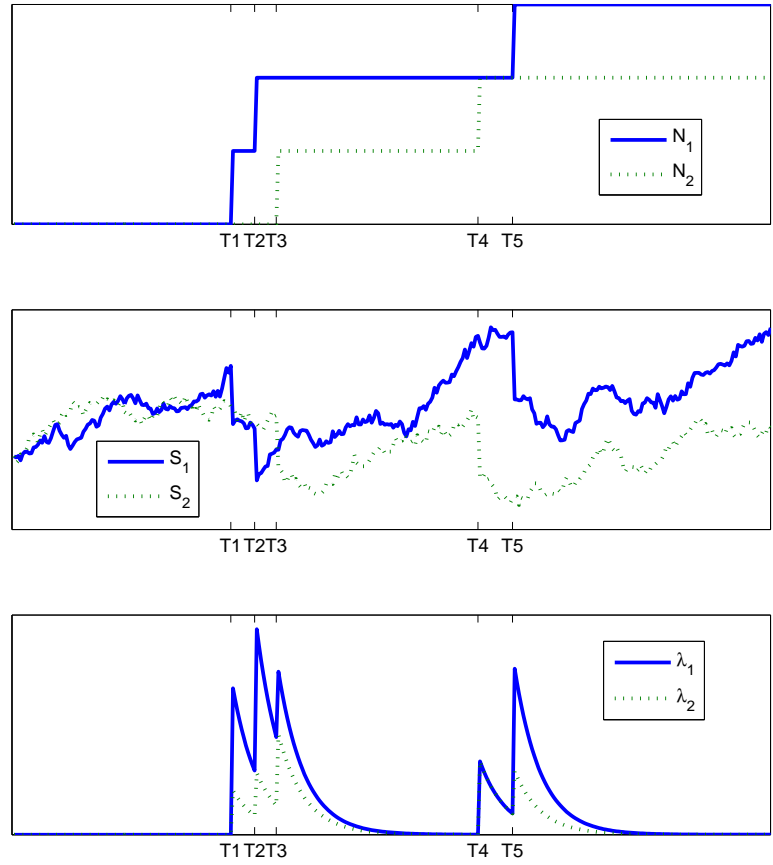}%
\caption{Cross- and self-excitation in a two asset-class world:\ Sample paths
of the jumps, asset prices and jump intensities.}%
\label{fig:2assetNSlambda}%
\end{center}
\end{figure}
%EndExpansion

We note that each compensated process $N_{l,t}-\int_{0}^{t}\lambda_{l,s}ds$ is
a local martingale. A number of important additional properties hold for this model:

\begin{enumerate}
\item \textbf{Markov Generator:\ } The Markov generator of this process acting
on differentiable functions $g:(\mathbb{Z}_{+})^{m}\times(\mathbb{R}_{+}%
)^{m}\to\mathbb{R}$ is given by
\[
[{\mathcal{A}} g](\boldsymbol{n},\boldsymbol{\lambda})=\sum_{l=1}^{m} \left[
\alpha_{l}\left(  \lambda_{l,\infty}-\lambda_{l}\right)  \frac{\partial
g}{\partial\lambda_{l}} +\lambda_{l}\left(  g(\boldsymbol{n}+\boldsymbol{e}%
_{l},\boldsymbol{\lambda}+\boldsymbol{d_{l}})-g(\boldsymbol{n},
\boldsymbol{\lambda})\right)  \right]
\]
where $\boldsymbol{e}_{l}=[\delta_{1l}, \dots, \delta_{ml}]^{\prime}$ and
$\boldsymbol{d}_{l}=[d_{1l}, \dots, d_{ml}]^{\prime}$. If
\[
\mathbb{E}\left[  \int^{t}_{0} \Bigl |\ [{\mathcal{A}}g](\boldsymbol{N}%
_{s},\boldsymbol{\lambda}_{s})-\sum_{\ell=1}^{m} \bigl[\alpha_{\ell}\left(
\lambda_{\ell,\infty}-\lambda_{\ell,s}\right)  \frac{\partial g(\boldsymbol{N}%
_{s},\boldsymbol{\lambda}_{s})}{\partial\lambda_{l}}\bigr]\ \Bigr| ds\right]
<\infty
\]
for all $t$ and $g$ is differentiable in $\lambda$, the Dynkin formula says
that for each $t\le T$ :
\begin{equation}
\label{Dynkin}\mathbb{E}[g(\boldsymbol{N}_{T},\boldsymbol{\lambda}%
_{T})|{\mathcal{F}}_{t}]=g(\boldsymbol{N}_{t},\boldsymbol{\lambda}_{t}) +
\mathbb{E}\left[  \int^{T}_{t} [{\mathcal{A}}g](\boldsymbol{N}_{s}%
,\boldsymbol{\lambda}_{s})\ ds|{\mathcal{F}}_{t}\right]  \ .
\end{equation}

\item \textbf{Stationarity Assumption:\ } Let the process $\boldsymbol{\lambda
}_{t}$ satisfy (\ref{lambdadynamics}) with $\boldsymbol{\lambda}_{0}%
\in\mathbb{R}_{+}^{m}$; $\alpha_{i}>0$; $d_{lj}\geq0$ and with
$\boldsymbol{\Gamma}=(\Gamma_{lj})_{l,j=1,\dots,m},$ where $\Gamma_{lj}%
=\alpha_{j}\delta_{lj}-d_{lj},$ a positive (hence invertible) matrix where
$\delta_{lj}$ is the Kronecker symbol. Then the intensities
$\boldsymbol{\lambda}$ are stationary processes with bounded moments. Using
the Dynkin formula, one can show that under this assumption the first moments
$f_{l}(t)=\mathbb{E}\left[  \lambda_{l}(t)\right]  $ converge as
$t\rightarrow\infty$ to the non-negative values $f_{l}(\infty)=\sum_{j=1}%
^{m}(\delta_{lj}-\alpha_{l}^{-1}d_{lj})\lambda_{j,\infty}$. Similar formulas
can be derived for the large time limits of the higher moment functions. Using
these bounds, the ergodic theorem for semimartingales (see e.g.
\cite{khasminskii60}) then implies that for any measurable function
$K(\lambda)$ that satisfies a bound $|K(\lambda)|\leq M(1+\Vert\lambda
\Vert^{2})$
\begin{equation}
\lim_{T\rightarrow\infty}\frac{1}{T}\int_{0}^{T}K(\lambda_{t})dt=\lim
_{t\rightarrow\infty}\mathbb{E}[K(\lambda_{t})]=\int_{\mathbb{R}_{+}^{n}%
}K(\lambda)\mu(d\lambda) \label{ergodic}%
\end{equation}
a.s. where $\mu(d\lambda)$ is the invariant (infinite time) measure of
$\lambda_{t}$.

\item \textbf{Affine Structure:\ } The joint characteristic function has the
affine form
\[
%\Phi(\boldsymbol{u},\boldsymbol{v})=
\mathbb{E}_{0,\lambda_{0}}[e^{i\boldsymbol{uN_{T}}+i\boldsymbol{v\lambda_{T}}%
}]=\exp\left[  iA(T;\boldsymbol{u},\boldsymbol{v})+i\sum\nolimits_{l=1}%
^{m}B_{l}(T;\boldsymbol{u},\boldsymbol{v})\lambda_{l,0}\right]
\]
where the deterministic functions $A,B_{l}$ satisfy the Riccati equations
\begin{align}
\frac{\partial A}{\partial T}  &  =\sum\nolimits_{l=1}^{m}\alpha_{l}%
\lambda_{l,\infty}B_{\ell};\quad A(0)=0;\\
\frac{\partial B_{l}}{\partial T}  &  =-\alpha_{l}B_{l}-i\left(
e^{iu_{l}+i\sum_{j=1}^{m}d_{jl}B_{j}}-1\right)  ;\quad B_{l}(0)=v_{l};\quad
l=1,\dots,m\text{ }.
\end{align}

\end{enumerate}

\section{Optimal Portfolio Selection When Jumps Are Mutually
Exciting\label{sec:choice}}

We now solve Merton's problem:\ the investor maximizes the expected utility of
consumption by investing in a set of $n$ risky assets and a riskless asset
over the infinite time horizon $t\in\lbrack0,\infty)$. The innovation is that
the risky assets are now subject to shocks generated by an $m$-dimensional
Hawkes jump-diffusion process.

\subsection{Asset Return Dynamics}

The riskless asset with price $S_{0,t}$ is assumed to earn a constant rate of
interest $r\geq0$. The $n$ risky assets with prices $\boldsymbol{S}%
_{t}=\left[  S_{1,t},\ldots,S_{n,t}\right]  ^{\prime}$ follow a semimartingale
dynamics with asset shocks generated by an $m$-dimensional Hawkes process.
Specifically, we assume
\begin{align}
\frac{dS_{0,t}}{S_{0,t}}  &  =rdt,\label{eq:S0}\\
\frac{dS_{i,t}}{S_{i,t-}}  &  =\left(  r+R_{i}\right)  dt+\sum_{j=1}^{n}%
\sigma_{i,j}dW_{j,t}+\sum_{l=1}^{m}J_{i,l}Z_{l,t}dN_{l,t},\text{ }i=1,...,n
\label{eq:dSi}%
\end{align}
Here $\boldsymbol{N}_{t}=$ $\left[  N_{1,t},\ldots,N_{m,t}\right]  ^{\prime}$
is an $m-$dimensional, $m\leq n,$ vector of mutually exciting Hawkes processes
with intensities $\boldsymbol{\lambda}_{t}=\left[  \lambda_{1,t},\dots
,\lambda_{m,t}\right]  ^{\prime}$ that follow the Markovian dynamics%
\begin{equation}
d\lambda_{l,t}=\alpha_{l}\left(  \lambda_{l,\infty}-\lambda_{l,t}\right)
dt+\sum\nolimits_{j=1}^{m}d_{lj}dN_{j,t},\quad l=1,\dots,m, \label{lambdaDE}%
\end{equation}
with constant parameters $\alpha_{l}>0,$ $\lambda_{l,\infty}\geq0,$ and
$d_{lj}\geq0$. Under the condition that $\boldsymbol{\Gamma}$ is a positive
matrix, the $\boldsymbol{\lambda}$ process is stationary.

The vector $\boldsymbol{W}_{t}=\left[  W_{1,t},\ldots,W_{n,t}\right]
^{\prime}$ is an $n-$dimensional standard Brownian motion. $J_{i,l}Z_{l,t}$ is
the response of asset $i$ to the $l$th shock where $Z_{l,t}$, a scalar random
variable with probability measure $\nu_{l}(dz)$ on $[0,1]$, is scaled on an
asset-by-asset basis by the deterministic scaling factor $J_{i,l}\in
\lbrack-1,0].$ For clarity, we chose to include only negative asset jumps in
the asset price dynamics, since those are the more relevant ones from both a
portfolio risk management perspective and their contribution to mutual excitation.

We assume that the individual Brownian motions, the Hawkes process and the
random variables $Z_{l}$ are mutually independent. The quantities $R_{i},$
$\sigma_{ij}$ and jump scaling factors $J_{i,l}$ are constant parameters. We
write $\boldsymbol{R}$ $\boldsymbol{=}$ $\left[  R_{1},\ldots,R_{n}\right]
^{\prime}$, $\boldsymbol{J}=(J_{i,l})_{i=1,...,n;l=1,\dots,m}$, and
$\boldsymbol{\sigma}=(\sigma_{i,j})_{i,j=1,\dots,n}$ and we assume that the
matrix $\boldsymbol{\Sigma}=\boldsymbol{\sigma\sigma}^{\prime}$ is nonsingular.

In Section \ref{sec:sector}, we will make further assumptions on the structure
of the matrix $\boldsymbol{\Sigma}$ to facilitate the derivation of an
explicit solution, assuming in particular that it possesses a factor
structure. But the existence and structure of the optimal portfolio solution
can be determined without further specialization, and we now turn to this problem.

\subsection{Wealth Dynamics and Expected Utility}

Let $\omega_{0,t}$ denote the percentage of wealth (or portfolio
weight)\ invested at time $t$ in the riskless asset and $\boldsymbol{\omega
}_{t}=\left[  \omega_{1,t},\ldots,\omega_{n,t}\right]  ^{\prime}$ denote the
vector of portfolio weights in each of the $n$ risky assets, assumed to be
adapted c\'{a}gl\'{a}d processes since the portfolio weights cannot anticipate
the jumps. The portfolio weights satisfy
\begin{equation}
\omega_{0,t}+%
%TCIMACRO{\dsum \nolimits_{i=1}^{n}}%
%BeginExpansion
{\displaystyle\sum\nolimits_{i=1}^{n}}
%EndExpansion
\omega_{i,t}=1\boldsymbol{.} \label{eq:omega}%
\end{equation}

The investor consumes $C_{t}$ at time $t$. In the absence of any income
derived outside his investments in these assets, the investor's wealth,
starting with the initial endowment $X_{0},$ follows the dynamics
\begin{align}
dX_{t}  &  =-C_{t}dt+\omega_{0,t}X_{t}\frac{dS_{0,t}}{S_{0,t-}}+%
%TCIMACRO{\dsum \nolimits_{i=1}^{n}}%
%BeginExpansion
{\displaystyle\sum\nolimits_{i=1}^{n}}
%EndExpansion
\omega_{i,t}X_{t}\frac{dS_{i,t}}{S_{i,t-}}\nonumber\\
&  =\left(  rX_{t}+\boldsymbol{\omega}_{t}^{\prime}\boldsymbol{R}X_{t}%
-C_{t}\right)  dt\text{ }\boldsymbol{+}\text{ }X_{t}\boldsymbol{\omega}%
_{t}^{\prime}\boldsymbol{\sigma}d\boldsymbol{W}_{t}+X_{t}\sum\nolimits_{l=1}%
^{m}\left(  \boldsymbol{\omega}_{t}^{\prime}\boldsymbol{J}\right)  _{l}%
Z_{l,t}dN_{l,t}. \label{eq:dX}%
\end{align}

We consider an investor with time-separable utility of consumption $U(\cdot)$
and subjective discount rate or \textquotedblleft impatience\textquotedblright%
\ parameter $\beta>0$. The investor's problem at any time $t\geq0$ is then to
pick the consumption and portfolio weight processes $\{C_{s}%
,\boldsymbol{\omega}_{s}\}_{t\leq s\leq\infty}$ which maximize the
infinite--horizon discounted expected utility of consumption. The optimal
policies $\{C_{s},\boldsymbol{\omega}_{s}\}_{t\leq s\leq\infty}$ are subject
to the \textit{admissibility condition} that the discounted wealth process
remains positive almost surely.

Stochastic dynamic programming (see, e.g., \cite{flemingsoner}) leads at time
$t$ to the discounted expected utility of consumption in the form
$V(X_{t},\boldsymbol{\lambda}_{t},t)$ where the value function is defined by%
\begin{equation}
V\left(  x,\boldsymbol{\lambda},t\right)  =\max_{\left\{  C_{s}%
,\boldsymbol{\omega}_{s};\text{ }t\leq s\leq\infty\right\}  }\mathbb{E}%
_{x,\boldsymbol{\lambda},t}\left[  \int_{t}^{\infty}e^{-\beta s}%
U(C_{s})ds\right]  \label{eq:V}%
\end{equation}
Here, the discounted wealth and intensities satisfy (\ref{eq:dX}) and
(\ref{lambdaDE}) over $[t,\infty)$ with initial conditions $X_{t}=x,$
$\boldsymbol{\lambda}_{t}=\boldsymbol{\lambda}$. Under the assumption that $V$
is sufficiently differentiable, the appropriate form of It\^{o}'s lemma (see,
e.g., \cite{protter2004}) for semi-martingale processes leads to the
Hamilton-Jacobi-Bellman equation that characterizes the optimal solution to
the investor's problem:%
\begin{align}
0  &  =\max_{\{C,\boldsymbol{\omega}\}}\left\{  \frac{\partial V\left(
x,\boldsymbol{\lambda},t\right)  }{\partial t}+\sum_{l=1}^{m}\alpha_{l}\left(
\lambda_{l,\infty}-\lambda_{l}\right)  \frac{\partial V\left(
x,\boldsymbol{\lambda},t\right)  }{\partial\lambda_{i}}+e^{-\beta
t}U(C)\right. \nonumber\\
&  +\frac{\partial V\left(  x,\boldsymbol{\lambda},t\right)  }{\partial
x}\left(  rx+\boldsymbol{\omega}^{\prime}\boldsymbol{R}x-C\right)  +\frac
{1}{2}\frac{\partial^{2}V\left(  x,\boldsymbol{\lambda},t\right)  }{\partial
x^{2}}\boldsymbol{\omega}^{\prime}\boldsymbol{\Sigma\omega}x^{2}%
+\label{eq:HJBforV}\\
&  \left.  \sum_{l=1}^{m}\lambda_{l}\int\left[  V\left(  x+\left(
\boldsymbol{\omega}^{\prime}\boldsymbol{J}\right)  _{l}zx,\boldsymbol{\lambda
}+\boldsymbol{d}_{l},t\right)  -V\left(  x,\boldsymbol{\lambda},t\right)
\right]  \nu_{l}\left(  dz\right)  \right\} \nonumber
\end{align}
with the transversality condition $\lim_{t\rightarrow\infty}\mathbb{E}\left[
V\left(  X_{t},\boldsymbol{\lambda}_{t},t\right)  \right]  =0.$

Using the standard time-homogeneity argument for infinite horizon problems, we
have that
\begin{align*}
e^{\beta t}V\left(  x,\boldsymbol{\lambda},t\right)   &  =\max_{\left\{
C_{s},\boldsymbol{\omega}_{s};\text{ }t\leq s\leq\infty\right\}  }%
\mathbb{E}_{x,\boldsymbol{\lambda},t}\left[  \int_{t}^{\infty}e^{-\beta
(s-t)}U(C_{s})ds\right] \\
&  =\max_{\left\{  C_{s},\boldsymbol{\omega}_{s};\text{ }t\leq s\leq
\infty\right\}  }\mathbb{E}_{x,\boldsymbol{\lambda},t}\left[  \int_{0}%
^{\infty}e^{-\beta u}U(C_{t+u})du\right] \\
&  =\max_{\left\{  C_{s},\boldsymbol{\omega}_{s};\text{ }0\leq s\leq
\infty\right\}  }\mathbb{E}_{x,\boldsymbol{\lambda},0}\left[  \int_{0}%
^{\infty}e^{-\beta u}U(C_{u})du\right] \\
&  =V(x,\boldsymbol{\lambda},0)\equiv L(x,\boldsymbol{\lambda})
\end{align*}
is independent of time. Thus $V\left(  x,\boldsymbol{\lambda},t\right)
=e^{-\beta t}L(x,\boldsymbol{\lambda})$ and (\ref{eq:HJBforV}) reduces to the
following time-independent equation for the value function $L:$%
\begin{align}
0=\max_{\left\{  C,\boldsymbol{\omega}\right\}  }  &  \left\{  U(C)-\beta
L(x,\boldsymbol{\lambda})+\sum_{l=1}^{m}\alpha_{l}\left(  \lambda_{l,\infty
}-\lambda_{l}\right)  \frac{\partial L\left(  x,\boldsymbol{\lambda}\right)
}{\partial\lambda_{l}}\right. \nonumber\\
&  +\frac{\partial L\left(  x,\boldsymbol{\lambda}\right)  }{\partial
x}\left(  rx+\boldsymbol{\omega}^{\prime}\boldsymbol{R}x-C\right)  +\frac
{1}{2}\frac{\partial^{2}L\left(  x,\boldsymbol{\lambda}\right)  }{\partial
x^{2}}\boldsymbol{\omega}^{\prime}\boldsymbol{\Sigma\omega}x^{2}%
\label{eq:HJBgeneral}\\
&  \left.  +\sum_{l=1}^{m}\lambda_{l}\int\left[  L\left(  x+\left(
\boldsymbol{\omega}^{\prime}\boldsymbol{J}\right)  _{l}zx,\boldsymbol{\lambda
}+\boldsymbol{d}_{l}\right)  -L\left(  x,\boldsymbol{\lambda}\right)  \right]
\nu_{l}\left(  dz\right)  \right\} \nonumber
\end{align}
with the transversality condition
\begin{equation}
\lim_{t\rightarrow\infty}\mathbb{E}\left[  e^{-\beta t}L\left(  X_{t}%
,\boldsymbol{\lambda}_{t}\right)  \right]  =0. \label{eq:Transversality}%
\end{equation}

The maximization problem in (\ref{eq:HJBgeneral}) separates into one for $C$,
with first order condition
\[
U^{\prime}\left(  C\right)  =\frac{\partial L\left(  x,\boldsymbol{\lambda
}\right)  }{\partial x}%
\]
and one for $\boldsymbol{\omega}:$%
\begin{align}
\boldsymbol{\omega}^{\ast}=\boldsymbol{\omega}^{\ast}(x,\boldsymbol{\lambda
}):=\operatorname{argmax}_{\left\{  \boldsymbol{\omega}\right\}  }  &
\left\{  \frac{\partial L\left(  x,\boldsymbol{\lambda}\right)  }{\partial
x}\boldsymbol{\omega}^{\prime}\boldsymbol{R}x+\frac{1}{2}\frac{\partial
^{2}L\left(  x,\boldsymbol{\lambda}\right)  }{\partial x^{2}}%
\boldsymbol{\omega}^{\prime}\boldsymbol{\Sigma\omega}x^{2}\right. \nonumber\\
&  \left.  +\sum_{l=1}^{m}\lambda_{l}\int\left[  L\left(  x+\left(
\boldsymbol{\omega}^{\prime}\boldsymbol{J}\right)  _{l}zx,\boldsymbol{\lambda
}+\boldsymbol{d}_{l}\right)  -L\left(  x,\boldsymbol{\lambda}\right)  \right]
\nu_{l}\left(  dz\right)  \right\}  \label{eq:HJM_omegapart}%
\end{align}
At time $t\geq0$, given wealth $X_{t}$ and intensity vector
$\boldsymbol{\lambda}_{t}$, the optimal consumption choice is therefore
$C_{t}^{\ast}=C^{\ast}(X_{t},\boldsymbol{\lambda}_{t})$ where
\begin{equation}
C^{\ast}(x,\boldsymbol{\lambda})\equiv\left[  U^{\prime}\right]  ^{-1}\left(
\partial L\left(  x,\boldsymbol{\lambda}\right)  /{\partial x}\right)  .
\label{eq:CstarGeneral}%
\end{equation}
In order to determine the optimal portfolio weights, wealth and value
function, we need to be more specific about the utility function $U.$

\section{Log-Utility Investors\label{sec:logU}}

There are three classic utility functions for which one may hope to make
further analytical progress, namely the log investor whose utility of
consumption is the logarithm function, the power investor and the exponential
investor. Collectively, these examples are known as HARA utilities. The
optimal investment problem for various simpler types of market dynamics with
these utilities lead to separable forms for the value functions. As we shall
now show in this and Section \ref{sec:powerandexp}, this separation property
extends to the Hawkes-diffusion model, albeit with some extra twists. In this
section, we concentrate on the log-investor, for whom we are able to prove
strong results on the existence and uniqueness of the optimal strategy.

\subsection{Optimal Investment with Log-Utility\label{ssec:log}}

We now specialize the problem to that faced by an investor with logarithmic
utility, $U\left(  x\right)  =\log\left(  x\right)  $. To start, we look for a
candidate solution to (\ref{eq:HJBgeneral}) in the form
\begin{equation}
L(x,\boldsymbol{\lambda})=f\left(  \boldsymbol{\lambda}\right)  +M^{-1}%
\log\left(  x\right)  \label{eq:Vform_log}%
\end{equation}
for some positive function $f$ and constant $M$. Then%
\begin{align}
\frac{\partial L\left(  x,\boldsymbol{\lambda}\right)  }{\partial\lambda_{l}}
&  =f_{\lambda_{l}}\left(  \boldsymbol{\lambda}\right)  ,\text{ }%
\frac{\partial L\left(  x,\boldsymbol{\lambda}\right)  }{\partial x}%
=M^{-1}x^{-1},\nonumber\\
\frac{\partial^{2}L\left(  x,\boldsymbol{\lambda}\right)  }{\partial x^{2}}
&  =-M^{-1}x^{-2}. \label{eq:deriveV_log}%
\end{align}
and the optimal policy for the portfolio weights at time $t\geq0$ is
$\boldsymbol{\omega}_{t}^{\ast}=\boldsymbol{\omega}^{\ast}\left(
\boldsymbol{\lambda}_{t}\right)  $ where
\begin{align}
\boldsymbol{\omega}^{\ast}\left(  \boldsymbol{\lambda}\right)   &
=\operatorname{argmin}_{\boldsymbol{\omega}}K_{l}(\boldsymbol{\omega
},\boldsymbol{\lambda})\nonumber\\
K_{l}(\boldsymbol{\omega},\boldsymbol{\lambda})  &  \equiv\left\{
-\boldsymbol{\omega}^{\prime}\boldsymbol{R}+\frac{1}{2}\boldsymbol{\omega
}^{\prime}\boldsymbol{\Sigma\omega}-\sum_{l=1}^{m}\lambda_{l}\int\log\left(
1+\left(  \boldsymbol{\omega}^{\prime}\boldsymbol{J}\right)  _{l}z\right)
\nu_{l}\left(  dz\right)  \right\}  . \label{eq:omegastar}%
\end{align}
We note that the convexity of $K_{l}$ implies this minimization has a unique
solution $\boldsymbol{\omega}^{\ast}\left(  \boldsymbol{\lambda}_{t}\right)  $
for any $\boldsymbol{\lambda}\in\mathbb{R}_{+}^{n}$. As for the optimal
consumption policy, from equation (\ref{eq:CstarGeneral}), and the facts that
$[U^{\prime}]^{-1}(y)=y^{-1}$ and $\partial L\left(  x,\boldsymbol{\lambda
}\right)  /\partial x=M^{-1}x^{-1},$ we obtain $M=\beta$ and
\begin{equation}
C_{t}^{\ast}=\beta X_{t} \label{eq:Cstar_withK}%
\end{equation}

Next, we substitute the optimal $\left(  C^{\ast},\boldsymbol{\omega}^{\ast
}\right)  $ into (\ref{eq:HJBgeneral}) and determine that the function $f$
must solve
\begin{equation}
\lbrack\mathcal{A}f](\boldsymbol{\lambda})-\beta f\left(  \boldsymbol{\lambda
}\right)  =F(\boldsymbol{\lambda}),\quad\boldsymbol{\lambda}\in\mathbb{R}%
_{+}^{n} \label{HJBf}%
\end{equation}
where the Markov generator $\mathcal{A}$ for the process $\boldsymbol{\lambda
}_{t}$ is given by
\begin{equation}
\lbrack\mathcal{A}f](\boldsymbol{\lambda})=\sum_{l=1}^{m}\left(  \alpha
_{l}\left(  \lambda_{l,\infty}-\lambda_{l}\right)  f_{\lambda_{l}}\left(
\boldsymbol{\lambda}\right)  +\lambda_{l}\left[  f\left(  \boldsymbol{\lambda
}+\boldsymbol{d}_{l}\right)  -f\left(  \boldsymbol{\lambda}\right)  \right]
\right)  \label{Generator}%
\end{equation}
and the nonhomogeneous term is
\[
F(\boldsymbol{\lambda})=1-\frac{r}{\beta}-\log\beta+\beta^{-1}K_{l}%
(\boldsymbol{\omega}^{\ast}(\boldsymbol{\lambda}),\boldsymbol{\lambda}).
\]

The following lemma gives a computable formula for the smooth solution of
(\ref{HJBf}).

\begin{lemma}
\label{FKlemma} The function $f:\mathbb{R}_{+}^{m}\rightarrow\mathbb{R}_{+}$
defined by (\ref{HJBf}) is differentiable and given by the absolutely
convergent integral
\begin{equation}
\label{tildeV}f(\boldsymbol{\lambda})=\int_{0}^{\infty}e^{-\beta s}%
\mathbb{E}_{0,\boldsymbol{\lambda}}\left[  F(\boldsymbol{\lambda}_{s})\right]
ds\ .
\end{equation}

\end{lemma}

\begin{proof}
By the ergodic property, we know both that $\mathbb{E}_{0,\boldsymbol{\lambda
}}\left[  F(\boldsymbol{\lambda}_{s})\right]  \rightarrow\mathbb{E}%
_{0,\boldsymbol{\lambda}}\left[  F(\boldsymbol{\lambda}_{\infty})\right]  $
and \newline$\frac{\partial}{\partial\lambda_{\ell}}\mathbb{E}%
_{0,\boldsymbol{\lambda}}\left[  F(\boldsymbol{\lambda}_{s})\right]
\rightarrow0$ as $s\to\infty$. One can also verify directly that there is a
constant $\tilde{M}>0$ such that
\begin{equation}
|F(\lambda)|\leq\tilde{M}(1+\Vert\lambda\Vert^{2})\ . \label{Fbound}%
\end{equation}
From these facts follows the absolute convergence both of the integral in
(\ref{tildeV}) and the integral
\[
\frac{\partial f}{\partial\lambda_{\ell}}=\int_{0}^{\infty}e^{-\beta s}%
\frac{\partial}{\partial\lambda_{\ell}}\mathbb{E}_{0,\boldsymbol{\lambda}%
}\left[  F(\boldsymbol{\lambda}_{s})\right]  ds\ .
\]
Since the right hand side of (\ref{tildeV}) is differentiable, the Feynman-Kac
formula implies it satisfies (\ref{HJBf}).
\end{proof}

\subsection{A Verification Result for the Log Investor\label{ssec:logverif}}

The following verification theorem follows the logic outlined in Section III.9
of \cite{flemingsoner} and ensures that the above argument correctly
characterizes both the optimal strategy and the associated value function. A
more general verification result for log investors can be found in
\cite{GollKallsen00}.

\begin{theorem}
Consider the optimal problem (\ref{eq:V}) for the log investor with impatience
parameter $\beta>0$, investing in the asset price model defined by
(\ref{eq:S0}), (\ref{eq:dSi}), (\ref{lambdaDE}) and satisfying the
Stationarity Assumption that the matrix $\Gamma=(\alpha_{j}\delta_{ij}%
-d_{ij})$ is positive.

\begin{enumerate}
\item The candidate solution $\tilde{V}(x,\boldsymbol{\lambda},t)=e^{-\beta
t}\left[  f(\boldsymbol{\lambda})+\beta^{-1}\log(x)\right]  $ is a classical
(i.e. differentiable) solution of the HJB equation (\ref{eq:HJBforV}).

\item For all initial conditions $x>0,\boldsymbol{\ \lambda}\in\mathbb{R}%
_{+}^{n}$, the pair of processes $(\boldsymbol{\omega}_{t}^{\ast},C_{t}^{\ast
}),t\geq0$ defined by (\ref{eq:omegastar}) and (\ref{eq:Cstar_withK}) is an
\textit{admissible policy}, in the sense that they are progressively
measurable and the process $X_{t}^{\ast}$ remains finite and positive
$(t,\omega)$ almost surely and solves the appropriate SDE.

\item Let $\mathcal{C}$ denote the class of admissible policies
$(\boldsymbol{\omega}_{t},C_{t}),t\geq0$ that satisfy
\begin{equation}
\lim_{t\rightarrow\infty}\mathbb{E}_{x,\boldsymbol{\lambda}}\left[  \tilde
{V}\left(  X_{t},\boldsymbol{\lambda}_{t},t\right)  \right]  \geq0 \label{TC+}%
\end{equation}
For any $(\boldsymbol{\omega},C)\in\mathcal{C}$,
\begin{equation}
\mathbb{E}_{x,\boldsymbol{\lambda}}\left[  \int_{0}^{\infty}e^{-\beta
s}U(C_{s})ds\right]  \leq\tilde{V}(x,\boldsymbol{\lambda},0).
\end{equation}

\item Let $V_{AS}$ denote the value function
\begin{equation}
V_{AS}\left(  x,\boldsymbol{\lambda}\right)  =\max_{(C,\boldsymbol{\omega}%
)\in\mathcal{C}}\mathbb{E}_{x,\boldsymbol{\lambda}}\left[  \int_{0}^{\infty
}e^{-\beta s}U(C_{s})ds\right]  \ . \label{eq:VAS}%
\end{equation}
The optimal policy $(\boldsymbol{\omega}_{t}^{\ast},C_{t}^{\ast}),t\geq0$
satisfies
\begin{equation}
\lim_{t\rightarrow\infty}\mathbb{E}_{x,\boldsymbol{\lambda}}\left[  \tilde
{V}\left(  X_{t}^{\ast},\boldsymbol{\lambda}_{t},t\right)  \right]  =0
\label{TC}%
\end{equation}
and the equality
\begin{equation}
\mathbb{E}_{x,\boldsymbol{\lambda}}\left[  \int_{0}^{\infty}e^{-\beta
s}U(C_{s}^{\ast})ds\right]  =\tilde{V}(x,\boldsymbol{\lambda},0).
\end{equation}
Hence $\tilde{V}(x,\boldsymbol{\lambda},0)=V_{AS}(x,\boldsymbol{\lambda})$ and
$(\boldsymbol{\omega}^{\ast},C^{\ast})$ is the optimal portfolio policy in the
class $\mathcal{C}$.
\end{enumerate}
\end{theorem}

\begin{proof}
That $\tilde V$ is a classical solution of (\ref{eq:HJBforV}) follows from
Lemma \ref{FKlemma} which shows that $f(\lambda)$ is a differentiable solution
of (\ref{HJBf}). That $(\boldsymbol{\omega}_{t}^{\ast},C_{t}^{\ast}),t\geq0$
is admissible follows by general considerations.

Suppose $(\boldsymbol{\omega},C)\in\mathcal{C}$ and consider the process
$\xi_{s}=\tilde{V}(X_{s},\boldsymbol{\lambda}_{s},s)$. For any $t>0$, the
Dynkin formula implies that
\begin{align}
\mathbb{E}_{x,\boldsymbol{\lambda}}\left[  \xi_{t}\right]   &  =\xi_{0}%
+\int_{0}^{t}\mathbb{E}_{x,\boldsymbol{\lambda}}\Biggl[-\beta\tilde{V}%
(X_{s},\boldsymbol{\lambda}_{s},s)+\sum_{l=1}^{m}\alpha_{l}\left(
\lambda_{l,\infty}-\lambda_{l,s}\right)  \frac{\partial\tilde{V}\left(
X_{s},\boldsymbol{\lambda}_{s},s\right)  }{\partial\lambda_{l}}\nonumber\\
&  +\frac{\partial\tilde{V}\left(  X_{s},\boldsymbol{\lambda}_{s},s\right)
}{\partial x}\left(  rX_{s}+\boldsymbol{\omega}_{s}^{\prime}\boldsymbol{R}%
X_{s}-C_{s}\right)  +\frac{1}{2}\frac{\partial^{2}\tilde{V}\left(
X_{s},\boldsymbol{\lambda}_{s},s\right)  }{\partial x^{2}}\boldsymbol{\omega
}_{s}^{\prime}\boldsymbol{\Sigma\omega}_{s}X_{s}^{2}\\
&  +\sum_{l=1}^{m}\lambda_{l,s-}\int\left[  \tilde{V}\left(  X_{s}+\left(
\boldsymbol{\omega}_{s}^{\prime}\boldsymbol{J}\right)  _{l}zX_{s}%
,\boldsymbol{\lambda_{s-}}+\boldsymbol{d}_{l},s\right)  -\tilde{V}\left(
X_{s},\boldsymbol{\lambda}_{s-},s\right)  \right]  \nu_{l}\left(  dz\right)
\Biggr]ds\nonumber
\end{align}
Using the fact that $\tilde{V}$ solves the HJB equation (\ref{eq:HJBforV})
leads in the usual way to the inequality%

\begin{equation}
\label{verificationInt}\mathbb{E}_{x,\boldsymbol{\lambda}}\left[  \xi
_{t}\right]  \leq\xi_{0}-\mathbb{E}_{x,\boldsymbol{\lambda}}\left[  \int
_{0}^{t}e^{-\beta s}U(C_{s})ds\right]
\end{equation}
Finally, one can use the transversality condition (\ref{TC+}) to take the
limit $t\rightarrow\infty$ and obtain the desired result
\[
\mathbb{E}_{x,\boldsymbol{\lambda}}\left[  \int_{0}^{\infty}e^{-\beta
s}U(C_{s})ds\right]  \leq\xi_{0}=\tilde{V}(x,\lambda,0).
\]

The strategy $(\boldsymbol{\omega}_{t}^{\ast},C_{t}^{\ast}),t\geq0$ maximizes
(\ref{eq:HJBforV}) for almost every $(t,\omega)$, which means
(\ref{verificationInt}) holds as an equality for every $t>0$. Finally, one
needs to verify (\ref{TC}) in order to conclude that $\xi_{0}=\tilde{V}%
_{0}=V_{AS}$. First, by the ergodic property (\ref{ergodic}) one has
\[
\lim_{t\rightarrow\infty}\mathbb{E}_{x,\boldsymbol{\lambda}}\left[  e^{-\beta
t}|f(\boldsymbol{\lambda}_{t})|\right]  =0
\]
Next, by plugging in the optimal consumption $C^{\ast}=\beta X^{\ast}$ one
then finds one can solve the problem of optimizing the expected utility of
terminal wealth over the interval $[0,t]$ with the adjusted interest rate
$\hat{r}=r-\beta$ to show that
\[
\mathbb{E}_{x,\boldsymbol{\lambda}}\left[  e^{-\beta t}\log(X_{t}^{\ast
})\right]  =[te^{-\beta t}]\cdot\frac{1}{t}\int_{0}^{t}\mathbb{E}%
_{x,\boldsymbol{\lambda}}[\hat{r}+K_{l}(\boldsymbol{\omega}^{\ast
}(\boldsymbol{\lambda_{s}}),\boldsymbol{\lambda_{s}})]ds
\]
From the bounds (\ref{Fbound}) on the functions $F,K$ this can be seen to
converge to $0$ as $t\rightarrow\infty$, again by the ergodic property
(\ref{ergodic}).
\end{proof}

\subsection{Properties of the Solution}

The log investor is often described as \textquotedblleft
myopic\textquotedblright\ because she acts at each moment in time as if the
dynamical variables are in fact static parameters. Thus we should not be
surprised that the optimal consumption and portfolio weights at any time given
by (\ref{eq:Cstar_withK}) and (\ref{eq:omegastar}) are independent of the
coefficients of the SDEs driving the asset returns. In particular, at any time
$t$, the policy $(\boldsymbol{\omega}_{t}^{\ast},C_{t}^{\ast}%
,\boldsymbol{\lambda}_{t})$ is precisely the same as the policy when the jump
frequency is treated as a constant $\lambda=\lambda_{t}$, although that
\textquotedblleft constant\textquotedblright\ is changed at each instant.

A second observation is that the policy $(\boldsymbol{\omega}_{t}^{\ast}%
,C_{t}^{\ast},\boldsymbol{\lambda}_{t}),t\geq0$ does not require knowledge of
the function $f(\lambda)$. However, determining the function $f$ requires
solving the non-homogeneous equation (\ref{HJBf}), or equivalently evaluating
the integral in (\ref{tildeV}). For this, we need some further structure on
the problem, which we now add in order to derive the complete form of the solution.

\section{Additional Structure on the Diffusive and Jump
Risks\label{sec:sector}}

To find interesting examples of closed form optimal portfolio solutions, it is
convenient to model the variance-covariance matrix $\Sigma$ of the diffusive
part of asset returns in such a way that its inverse is explicit. For this
purpose, we adopt a modelling approach that is common in asset pricing, namely
to assume a factor structure. That is, we specify a block-structure for
$\Sigma$ consisting of $k$ blocks of dimension $m,$ with $n=mk$:
\begin{equation}
\underset{n\times n}{\boldsymbol{\Sigma}}=\boldsymbol{\sigma\sigma}^{\prime
}=\left(
\begin{array}
[c]{ccc}%
\Sigma_{1,1} & \Sigma_{1,2} & \cdots\\
\Sigma_{2,1} & \ddots & \Sigma_{2,m}\\
\cdots & \Sigma_{m,m-1} & \Sigma_{m,m}%
\end{array}
\right)  \label{eq:Sigma_msector}%
\end{equation}
with diagonal (or within-asset class) blocks%
\begin{equation}
\underset{k\times k}{\boldsymbol{\Sigma}_{l,l}}=\upsilon_{l}^{2}\left(
\begin{array}
[c]{ccc}%
1 & \rho_{l,l} & \cdots\\
\rho_{l,l} & \ddots & \rho_{l,l}\\
\cdots & \rho_{l,l} & 1
\end{array}
\right)  \label{eq:Sigmawithin}%
\end{equation}
and off-diagonal (or across-asset class) blocks%
\begin{equation}
\underset{k\times k}{\boldsymbol{\Sigma}_{l,s}}=0 \label{eq:Sigmaacross}%
\end{equation}
where $1>\rho_{l,l}>-1/(k-1).$

The spectral decomposition of the $\Sigma$ matrix is
\begin{equation}
\boldsymbol{\Sigma}\text{ \ }=\underset{=\text{ \ }\bar{\Sigma}}{\text{
\ }\underbrace{\sum\nolimits_{l=1}^{m}\kappa_{1l}\frac{1}{k}\boldsymbol{1}%
_{l}\boldsymbol{1}_{l}^{\prime}}}\text{ \ }+\text{ \ }\underset{=\text{
\ }\Sigma^{\perp}}{\underbrace{\sum\nolimits_{l=1}^{m}\kappa_{2l}\left(
\boldsymbol{M}_{l}\boldsymbol{-}\frac{1}{k}\boldsymbol{1}_{l}\boldsymbol{1}%
_{l}^{\prime}\right)  }} \label{eq:Sigma_msector_spectral}%
\end{equation}
where%
\begin{align}
\kappa_{1l}  &  =v_{l}^{2}\left(  1+\left(  k-1\right)  \rho_{l,l}\right)
\label{eq:kappa1l_msector}\\
\kappa_{2l}  &  =v_{l}^{2}\left(  1-\rho_{l,l}\right)
\label{eq:kappa2l_msector}%
\end{align}
are the $2m$ distinct eigenvalues of $\Sigma$. The multiplicity of each
$\kappa_{1l}$ is $1$, and the multiplicity of each $\kappa_{2l}$ is $k-1$. The
eigenvector for $\kappa_{1l}$ is $\boldsymbol{1}_{l}$, the $n-$vector with
ones placed in the $k$ rows corresponding to the $l-$block and zeros
everywhere else, that is%
\begin{equation}
\boldsymbol{1}_{l}=[0,\ldots,0,\underset{\text{asset class }l}{\underbrace
{1,\ldots,1}},0,\ldots,0]^{\prime}, \label{eq:1l}%
\end{equation}
where the first $1$ is located in the $k\left(  l-1\right)  +1$ coordinate.
$\boldsymbol{M}_{l}$ is an $n\times n$ block diagonal matrix with a $k\times
k$ identity matrix $\boldsymbol{I}_{k}$ placed in the $l-$block and zeros
everywhere else:%
\begin{equation}
\underset{n\times n}{\boldsymbol{M}_{l}}\text{ \ }=\left(
\begin{array}
[c]{ccc}%
0 & \cdots & 0\\
\vdots & \boldsymbol{I}_{k} & \vdots\\
0 & \cdots & 0
\end{array}
\right)  , \label{eq:Fl}%
\end{equation}
Corresponding to the above spectral structure, we have the orthogonal
decomposition $\mathbb{R}^{n}=\bar{V}\oplus V^{\perp}$ where $\bar{V}$ is the
span of $\{\boldsymbol{1}_{l}\}_{l=1,..,m}$ and $V^{\perp}$ is the orthogonal space.

As for the vector $\boldsymbol{J}$ of jump amplification coefficients, we
assume that
\begin{equation}
\underset{n\times m}{\boldsymbol{J}}\ \boldsymbol{=}\ [\boldsymbol{J}%
_{1},...,\boldsymbol{J}_{m}]=\left(
\begin{array}
[c]{ccc}%
J_{1,1} & \cdots & J_{1,m}\\
\vdots & \ddots & \vdots\\
J_{n,1} & \cdots & J_{n,m}%
\end{array}
\right)  \label{eq:Jmat}%
\end{equation}
where
\begin{equation}
\boldsymbol{J}_{l}\text{ \ }=\text{ }j_{l}\boldsymbol{1}_{l}\text{ \ }=\text{
\ }[\underset{\text{asset class }1}{\underbrace{0,\ldots,0}},\ldots
,\underset{\text{asset class }l}{\underbrace{j_{l},\ldots,j_{l}}}%
,\ldots\underset{\text{asset class }m}{\underbrace{,0,\ldots,0}}]^{\prime}
\label{eq:J_msector}%
\end{equation}
for $l=1,...,m.$ This structure means that assets within a given class $l$
have the same response to the arrival of a jump, i.e., to a change in
$\boldsymbol{N}_{t}$. But the proportional response $j_{l}$ of assets of
different classes to the arrival of a jump can be different ($j_{l}\neq j_{h}$
for $l\neq h$).

Finally, we assume that the vector of expected excess returns has the form
\begin{equation}
\boldsymbol{R}\text{ \ }\boldsymbol{=}\text{ \ }\sum\nolimits_{l=1}^{m}\bar
{R}_{l}\boldsymbol{1}_{l}\text{ }\boldsymbol{+}\text{ }\boldsymbol{{R}^{\perp
}}\text{ \ }\boldsymbol{=\ \bar{R}+{R}^{\perp}.} \label{eq:R_msector}%
\end{equation}
Here, we allow the expected excess returns to differ both within and across
asset classes, by allowing $\boldsymbol{{R}^{\perp}\neq0}$ . The components of
$\boldsymbol{R}$ play the role of the assets' alphas. The general
$\boldsymbol{{R}^{\perp}}$ is orthogonal to each $\boldsymbol{1}_{l}$ and has
the form
\[
\boldsymbol{{R}^{\perp}=[R}_{1}^{\perp\prime},...,\boldsymbol{R}_{m}%
^{\perp\prime}]^{\prime}%
\]
where each of the $k$-vectors $\boldsymbol{R}_{l}^{\perp}$ is orthogonal to
the $k$-vector $\boldsymbol{1}$.

\subsection{Closed-Form Optimal Portfolio Solution}

We now look for a vector of optimal portfolio weights $\boldsymbol{\omega,}$
and it is convenient to look for it in the form of its decomposition using the
same basis as above,
\begin{equation}
\boldsymbol{\omega}\text{ \ }\boldsymbol{=}\text{ \ }\sum\nolimits_{l=1}%
^{m}\bar{\omega}_{l}\boldsymbol{1}_{l}\text{ }\boldsymbol{+}\text{
}\boldsymbol{{\omega}^{\perp}=\bar{\omega}+{\omega}^{\perp}}
\label{eq:Omega_msector}%
\end{equation}
where $\boldsymbol{\omega}^{\perp}=\boldsymbol{[\omega}_{1}^{\perp\prime
},...,\boldsymbol{\omega}_{l}^{\perp\prime}]^{\prime}.$ The objective function
$K_{2}(\boldsymbol{\omega})$ in (\ref{eq:omegastar}) to be minimized reduces
to
\begin{align*}
K_{2}(\boldsymbol{\omega})  &  =\left\{  -\boldsymbol{\omega}^{\prime
}\boldsymbol{R}+\frac{1}{2}\boldsymbol{\omega}^{\prime}\boldsymbol{\Sigma
\omega}-\sum_{l=1}^{n}\lambda_{l}\int\log\left(  1+\left(  \boldsymbol{\omega
}^{\prime}\boldsymbol{J}\right)  _{l}z\right)  \nu_{l}\left(  dz\right)
\right\} \\
&  =-\boldsymbol{\omega}^{\perp\prime}\boldsymbol{R}^{\perp}-\sum
\nolimits_{l=1}^{m}\bar{\omega}_{l}\bar{R}_{l}\boldsymbol{1}_{l}^{\prime
}\boldsymbol{1}_{l}\\
&  +\frac{1}{2}\boldsymbol{\omega}^{\perp\prime}\Sigma^{\perp}%
\boldsymbol{\omega}^{\perp}+\frac{1}{2}\sum\limits_{l=1}^{m}\bar{\omega}%
_{l}^{2}\kappa_{1l}\frac{1}{k}\boldsymbol{1}_{l}^{\prime}\boldsymbol{1}%
_{l}\boldsymbol{1}_{l}^{\prime}\boldsymbol{1}_{l}\\
&  -\sum_{l=1}^{n}\lambda_{l}\int\log\left(  1+k\bar{\omega}_{l}j_{l}z\right)
\nu_{l}\left(  dz\right)  .
\end{align*}
The minimization problem then separates as%
\begin{equation}
\left(  \boldsymbol{\omega}^{\perp\ast},\boldsymbol{\bar{\omega}}^{\ast
}\right)  =\operatorname{argmin}_{\left\{  \boldsymbol{\omega}^{\perp
},\boldsymbol{\bar{\omega}}\right\}  }\left\{  g^{\perp}(\boldsymbol{\omega
}^{\perp})+\bar{g}(\bar{\omega})\right\}  \label{eq:omegaseparate_msector}%
\end{equation}
where%
\begin{align}
g^{\perp}(\boldsymbol{\omega}^{\perp})  &  =-\boldsymbol{\omega}^{\perp\prime
}\boldsymbol{R}^{\perp}+\frac{1}{2}\boldsymbol{\omega}^{\perp\prime}%
\Sigma^{\perp}\boldsymbol{\omega}^{\perp}\label{eq:gperp_msector}\\
\bar{g}(\boldsymbol{\bar{\omega}})  &  =-k\sum\limits_{l=1}^{m}\bar{\omega
}_{l}\bar{R}_{l}+\frac{1}{2}k\sum\limits_{l=1}^{m}\bar{\omega}_{l}^{2}%
\kappa_{1l}-\sum_{l=1}^{m}\lambda_{l,t}\int\log\left(  1+k\bar{\omega}%
_{l}j_{l}z\right)  \nu_{l}\left(  dz\right)  . \label{eq:gbar_msector}%
\end{align}

For the first part, minimizing $g^{\perp}(\boldsymbol{\omega}^{\perp}),$ the
structure of $\Sigma^{\perp}$ implies that
\begin{align*}
g^{\perp}(\boldsymbol{\omega}^{\perp})  &  =-\boldsymbol{\omega}^{\perp\prime
}\boldsymbol{R}^{\perp}+\frac{1}{2}\sum\nolimits_{l=1}^{m}\kappa
_{2l}\boldsymbol{\omega}^{\perp\prime}\left(  \boldsymbol{M}_{l}%
\boldsymbol{-}\frac{1}{k}\boldsymbol{1}_{l}\boldsymbol{1}_{l}^{\prime}\right)
\boldsymbol{\omega}^{\perp}\\
&  =-\boldsymbol{\omega}^{\perp\prime}\boldsymbol{R}^{\perp}+\frac{1}{2}%
\sum\nolimits_{l=1}^{m}\kappa_{2l}\boldsymbol{\omega}^{\perp\prime
}\boldsymbol{M}_{l}\boldsymbol{\omega}^{\perp}\\
&  =-\sum\nolimits_{l=1}^{m}\boldsymbol{\omega}_{l}^{\perp\prime
}\boldsymbol{R}_{l}^{\perp}+\frac{1}{2}\sum\nolimits_{l=1}^{m}\kappa
_{2l}\boldsymbol{\omega}_{l}^{\perp\prime}\boldsymbol{\omega}_{l}^{\perp}%
\end{align*}
and therefore the optimal solution $\boldsymbol{\omega}^{\perp\ast}$ has
blocks
\begin{equation}
\boldsymbol{\omega}_{l}^{\perp\ast}=\frac{1}{\kappa_{2l}}\boldsymbol{R}%
_{l}^{\perp} \label{eq:omega_perp_star}%
\end{equation}
for $l=1,\ldots,m.$ This part of the solution depends only on the diffusive
characteristics (expected returns and variance-covariance)\ of the asset returns.

The problem of minimizing $\bar{g}(\boldsymbol{\bar{\omega}})$ separates
itself into $m$ separate minimization problems. With the change of variable
\begin{equation}
\varpi_{ln}=k\bar{\omega}_{l} \label{eq:changeofvar}%
\end{equation}
we see that%
\begin{equation}
\varpi_{ln}^{\ast}=\operatorname{argmin}_{\left\{  \varpi_{ln}\right\}
}\left\{  -\varpi_{ln}R_{l}+\frac{1}{2}\varpi_{ln}^{2}\kappa_{1l}%
/k-\lambda_{l,t}\int\log\left(  1+\varpi_{ln}j_{l}z\right)  \nu_{l}\left(
dz\right)  \right\}  \label{eq:ystarn_msector}%
\end{equation}
The convexity of the objective function implies the existence of the
minimizer. We can then determine $\varpi_{ln}^{\ast}$ in closed form.

Two cases are explicitly solvable. We first consider the case where the jump
size is deterministic. In this situation, $\nu_{l}(dz)=\delta\left(  z=\bar
{z}_{l}\right)  $ and the objective functions become:%
\begin{align}
f_{n}\left(  \varpi\right)   &  =-\sum\nolimits_{l=1}^{m}\varpi_{ln}\bar
{R}_{l}+\frac{1}{2}\sum\nolimits_{l=1}^{m}\varpi_{ln}^{2}\kappa_{1l}%
/k\nonumber\\
&  -\sum_{l=1}^{m}\lambda_{l,t}\log\left(  1+\varpi_{ln}j_{l}\bar{z}%
_{l}\right)  . \label{eq:objmsector_example}%
\end{align}
The first order conditions for the asset allocation parameters $\varpi_{ln}$
are given by%
\begin{equation}
-\bar{R}_{l}+\varpi_{ln}\kappa_{1l}/k-\lambda_{l,t}j_{l}\bar{z}_{l}\left(
1+\varpi_{ln}j_{l}\bar{z}_{l}\right)  ^{-1}=0\text{ for }l=1,\ldots,m.
\label{eq:FOC_Log_msector}%
\end{equation}
These first order conditions form a system of $m$ independent quadratic
equations. Each separate equation (\ref{eq:FOC_Log_msector}) admit a unique
solution $\varpi_{ln}$ satisfying the solvency constraint $\varpi_{ln}%
j_{l}\bar{z}_{l}>-1$. These are solvable in closed form:%
\begin{equation}
\bar{\omega}_{l}^{\ast}=\frac{\varpi_{ln}^{\ast}}{k}=\frac{-\kappa
_{1l}/k+j_{l}\bar{z}_{l}\bar{R}_{l}+\sqrt{\left(  j_{l}\bar{z}_{l}\bar{R}%
_{l}+\kappa_{1l}/k\right)  ^{2}+4\lambda_{l,t}j_{l}^{2}\bar{z}_{l}^{2}%
\kappa_{1l}/k}}{2j_{l}\bar{z}_{l}\kappa_{1l}/k^{2}}.
\label{eq:closed_form_omega_msector}%
\end{equation}

A second case that is solvable in closed form is one where each jump term
$Z_{l,t}$ has a binomial distribution ($u_{l}$ with probability $p_{l}$ or
$d_{l}$ with probability $1-p_{l}$), since the corresponding first-order
conditions are cubic. The first order conditions are obtained by
differentiation with respect to $\varpi_{ln}$ of the objective function stated
in (\ref{eq:ystarn_msector}), namely:%
\begin{equation}
-\bar{R}_{l}+\varpi_{ln}\kappa_{1l}/k-\lambda_{l,t}j_{l}\int\left(
1+\varpi_{ln}j_{l}z\right)  ^{-1}z\nu_{l}\left(  dz\right)  =0\text{ for
}l=1,\ldots,m. \label{eq:generaljump_FOC}%
\end{equation}
For binomially-distributed jumps, the conditions reduce to%
\begin{equation}
-R_{l}+\varpi_{ln}\kappa_{1l}/k-\lambda_{l,t}j_{l}\left(  p_{l}u_{l}\left(
1+\varpi_{ln}j_{l}u_{l}\right)  ^{-1}+\left(  1-p_{l}\right)  d_{l}\left(
1+\varpi_{ln}j_{l}d_{l}\right)  ^{-1}\right)  =0
\label{eq:FOC_Log_msector_binomial}%
\end{equation}
which produce a cubic polynomial equation in $\varpi_{ln},$ again explicitly
solvable, separately for each for $l=1,\ldots,m$.

\section{Consequences for the Optimal Portfolio Allocation\label{sec:conseq}}

We now investigate in more detail the consequences of the explicit portfolio
weight formulae for an optimal asset allocation. The first element we note
from (\ref{eq:closed_form_omega_msector}) is the fact that the optimal
portfolio is time-varying, since its composition changes with the jump
intensities $\lambda_{lt}.$

In the case of purely diffusive risk, the optimal portfolio weights reduce to
the classical Merton formula
\begin{equation}
\boldsymbol{\omega}^{\ast}=\boldsymbol{\Sigma}^{-1}\boldsymbol{R}%
\end{equation}
or, replacing $\boldsymbol{\Sigma}^{-1}$ by its explicit expression
\begin{equation}
\left\{
\begin{array}
[c]{l}%
\boldsymbol{\omega}_{l}^{\perp\ast}=\frac{1}{\kappa_{2l}}\boldsymbol{R}%
_{l}^{\perp}\\
\bar{\omega}_{l}^{\ast}=\frac{1}{\kappa_{1l}}\bar{R}_{l}%
\end{array}
\right.
\end{equation}
for $l=1,..,m.$ As is well known, the solution in this case is constant.

In the case where Poissonian jumps are added to the model, the solution
specializes to (\ref{eq:closed_form_omega_msector}) but with $\lambda_{l,t}$
replaced by the constant Poissonian jump intensity. As in the purely-diffusive
case, the solution becomes constant. A solution with similar qualitative
features would be obtained in the mutually exciting case if one replaced each
stochastic jump intensity by its unconditional expected value, although given
that the portfolio weights are nonlinear functions of $\lambda_{l,t}$ these
would not be the unconditional expected values of the portfolio weights.

Let us now return to the full solution in the mutually exciting case. A
univariate model captures only part of the mutual excitation phenomenon: with
a single asset, only time series self-excitation can take place. In order to
investigate the full potential impact of mutual excitation on optimal
portfolio holdings, we now specialize the results above to a two-asset model
where both time series and cross-sectional excitation can arise. Assets $1$
and $2$ can excite each other, not necessarily in a symmetric fashion,
depending on the $2\times2$ matrix of mutually exciting intensities with
coefficients $d_{ij},$ $i,j=1,2$. The formulae above specialize with $n=2,$
$k=1$ and $m=2,$ in which case $\kappa_{1l}=v_{l}^{2}$ for $l=1,2$ and hence:%
\begin{equation}
\left\{
\begin{array}
[c]{l}%
\boldsymbol{\omega}_{l}^{\perp\ast}=\frac{1}{\kappa_{2l}}\boldsymbol{R}%
_{l}^{\perp}\\
\bar{\omega}_{l}^{\ast}=\frac{-v_{l}^{2}+j_{l}\bar{z}_{l}R_{l}+\sqrt{\left(
j_{l}\bar{z}_{l}R_{l}+v_{l}^{2}\right)  ^{2}+4\lambda_{l,t}j_{l}^{2}\bar
{z}_{l}^{2}v_{l}^{2}}}{2j_{l}\bar{z}_{l}v_{l}^{2}}%
\end{array}
\right.  \label{eq:closedform2assets}%
\end{equation}

Consider the change in the optimal portfolio allocation of an investor who
observes a first shock, say to asset $1.$ For concreteness, let us return to
the two-asset class scenario illustrated in Figure \ref{fig:2assetNSlambda}.
In a Poissonian jump model, observing the first shock at time $T_{1}$ does not
change the investor's optimal portfolio: since jumps' future arrivals are
independent of past jumps, there is nothing to do going forward other than to
absorb the losses from the first jump. In the mutually exciting model,
however, the occurrence of the first jump at time $T_{1}$ self-excites the
jump intensity $\lambda_{1,t}$ for $t>T_{1}$. This increase in $\lambda_{1,t}$
translates, if $\bar{z}_{1}<0,$ into a reduced asset allocation to asset $1.$
Moreover, these jumps have a contagious effect on $S_{2}$ since a jump in
asset $1$ cross-excites the jump intensity $\lambda_{2,t}$ of asset $2$. If
$\bar{z}_{2}<0,$ then the optimal policy is to reduce the asset allocation to
asset $2.$ Note that the reduction to the position in both risky assets occurs
immediately after $T_{1},$ without waiting for future jumps.%

%TCIMACRO{\FRAME{ftphFU}{5.1224in}{5.3366in}{0pt}{\Qcb{Mutual excitation in a
%two asset-class world:\ Jump intensities (top panel) and optimal portfolio
%weights (bottom panel).}}{\Qlb{fig:2assetlambdaomega}}%
%{ACH_fig2assetlambdaomega.eps}{\special{ language "Scientific Word";
%type "GRAPHIC";  maintain-aspect-ratio TRUE;  display "USEDEF";
%valid_file "F";  width 5.1224in;  height 5.3366in;  depth 0pt;
%original-width 7.2774in;  original-height 7.5844in;  cropleft "0";
%filename 'ACH_fig2assetlambdaomega.eps';file-properties "XNPEU";}}}%
%BeginExpansion
\begin{figure}
[pth]
\begin{center}
\includegraphics[
height=5.3366in,
width=5.1224in
]%
{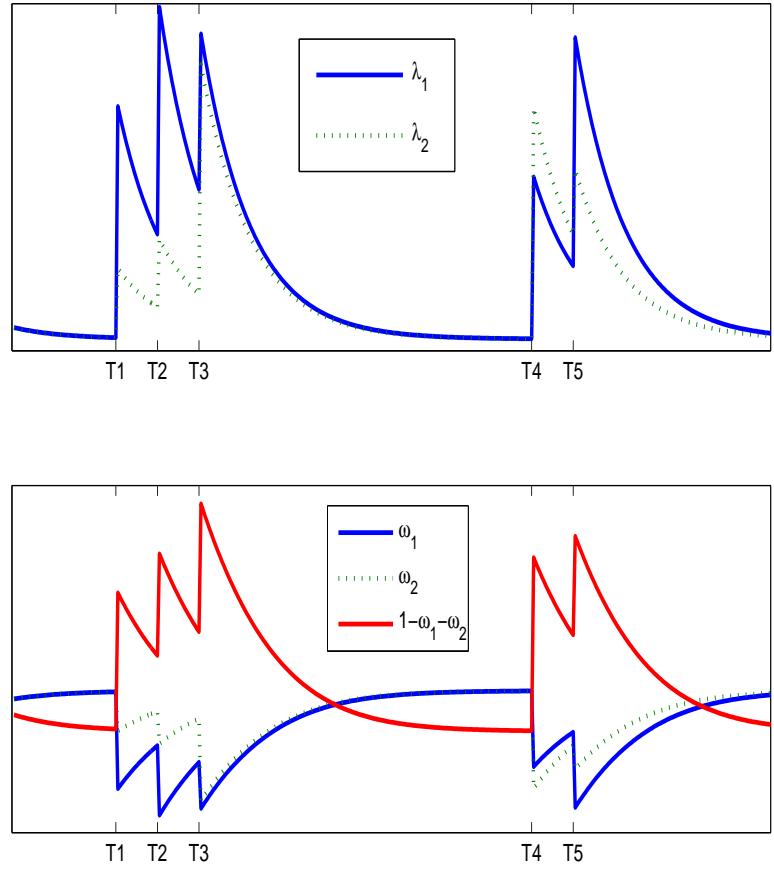}%
\caption{Mutual excitation in a two asset-class world:\ Jump intensities (top
panel) and optimal portfolio weights (bottom panel).}%
\label{fig:2assetlambdaomega}%
\end{center}
\end{figure}
%EndExpansion

For the same sample paths as in Figure \ref{fig:2assetNSlambda}, Figure
\ref{fig:2assetlambdaomega} shows the optimal portfolio weights. This is a
flight to quality, in the sense that the occurrence of a single jump in asset
$1$ causes the investor to flee both risky assets (or all of them in the
general $n$ case) for the safety of the riskless asset. This phenomenon is
well documented as an empirical reality in practical situations; we believe
that this is the first portfolio choice model to actually capture it in a
theoretical setting.

Increases in the jump intensities raise the probability of observing another
jump in $S_{1}$ at the future time $T_{2}$. This, in turn, raises the
probability of seeing a jump in $S_{2}$ at time $T_{3}$. Later on, at time
$T_{4}$, the jump in $S_{2}$ raises the probability of seeing a jump in
$S_{1}$ at some future time $T_{5}$, and so on. Mean reversion in the jump
intensities at respective rates $\alpha_{1}$ and $\alpha_{2}$ counteracts
these successive increases, keeping the intensities non-explosive (stationary,
in fact). The optimal portfolio policy reacts to each change in jump intensity
accordingly: increases lead to reduced asset allocation, decreases to
increased asset allocation. Inevitably, this analysis is conducted in a
partial equilibrium framework: it assumes among other things that expected
returns do not change as jump intensities change, or at least not sufficiently
to reverse the result.

Another interesting empirical phenomenon that can be revisited in light of
these optimal portfolio policies is home bias. Home bias refers to the
observed tendency of most investors' portfolios to be insufficiently
diversified internationally. In the model, the benefits from diversification
are much less valuable than in the standard diffusive model since
international assets do not protect as much against jumps in domestic assets
in the presence of cross-sectional mutual excitation.

Finally, one phenomenon that is often documented in financial crises is the
large increase in correlations between asset classes, with all of them
increasing towards $1.$ In the context of the model, empirical correlations
measured over a period where mutual excitation occurs will indeed be close to
$1$ as long as the jumps that result from mutual excitation are of the same
sign (say both $\bar{z}_{1}<0$ and $\bar{z}_{2}<0$), at least on average. In
such periods, the jumps' contribution to the observed correlation trumps the
continuous contribution.

%\subsection{Jump Immunization and Long/Short Strategies}

\section{Other HARA Investors\label{sec:powerandexp}}

The cases of power and exponential utility also lead to candidate value
functions and optimal portfolios in separable form. Due to the complexity of
the underlying HJB conditions, the relevant verification result requires a
lengthy, and perhaps uninformative, analysis that we have not yet completed.
Nonetheless, as we now show, the candidate solutions can be characterized in
terms of a fixed point problem that in principle can be solved numerically.

\subsection{Power Utility}

In this section, we consider the power investor with $U(c)=c^{\gamma}/\gamma$,
$\gamma\in(-\infty,0)\cup(0,1)$. The analysis now consists in verifying the
consistency of the following form for the solution to (\ref{eq:HJBgeneral}) in
the form
\[
L(x,\boldsymbol{\lambda})=x^{\gamma}g(\boldsymbol{\lambda})/\gamma
\]
for some positive function $g$. Substitution into (\ref{eq:HJBgeneral}) leads
to
\begin{align}
0=\max_{\left\{  C,\boldsymbol{\omega}\right\}  }  &  \left\{  U(C)-\beta
L\right. \\
&  +\frac{L}{g}\sum_{l=1}^{m}\left[  \alpha_{l}\left(  \lambda_{l,\infty
}-\lambda_{l}\right)  \frac{\partial g\left(  \boldsymbol{\lambda}\right)
}{\partial\lambda_{l}}+\lambda_{l}\left(  g(\boldsymbol{\lambda}%
+\boldsymbol{d}_{l})-g(\boldsymbol{\lambda})\right)  \right] \nonumber\\
&  +\gamma L\left(  r+\boldsymbol{\omega}^{\prime}\boldsymbol{R}-C/x\right)
+\frac{\gamma(\gamma-1)L}{2}\boldsymbol{\omega}^{\prime}\boldsymbol{\Sigma
\omega}\label{eq:HJBpower}\\
&  \left.  +L\sum_{l=1}^{m}\lambda_{l}g(\boldsymbol{\lambda}+\boldsymbol{d}%
_{l})/g(\boldsymbol{\lambda})\int\left[  \left(  1+\boldsymbol{\omega}%
^{\prime}\boldsymbol{J}z\right)  ^{\gamma}-1\right]  \nu_{l}\left(  dz\right)
\right\} \nonumber
\end{align}
The optimal policy for the portfolio weight at time $t\geq0$ is
$\boldsymbol{\omega}_{t}^{\ast}=\boldsymbol{\omega}^{\ast}(\boldsymbol{h}%
(\boldsymbol{\lambda}_{t}))$ where
\begin{align}
\boldsymbol{\omega}^{\ast}(\boldsymbol{\lambda})  &  =\left\{
\begin{array}
[c]{ll}%
\arg\min_{\boldsymbol{\omega}}K^{\gamma}(\boldsymbol{\omega}%
,\boldsymbol{\lambda}) & \gamma>0\\
\arg\max_{\boldsymbol{\omega}}K^{\gamma}(\boldsymbol{\omega}%
,\boldsymbol{\lambda}) & \gamma<0
\end{array}
\right. \\
{h_{l}}(\boldsymbol{{\lambda}})  &  ={\lambda_{l}}g(\boldsymbol{{\lambda}%
}+{\boldsymbol{{d}}_{l}})/g(\boldsymbol{{\lambda}})\\
{K^{\gamma}}(\boldsymbol{{\omega}},\boldsymbol{{\lambda}})  &  \equiv
-\gamma\boldsymbol{{\omega^{\prime}}}R-\frac{{\gamma(\gamma-1)}}%
{2}\boldsymbol{{\omega^{\prime}}}\Sigma\boldsymbol{{\omega}}-\sum_{l=1}%
^{m}\lambda_{l,t}{\int_{(0,1]}}\left[  {{{\left(  {1+\boldsymbol{{\omega
^{\prime}J}}z}\right)  }^{\gamma}}-1}\right]  \nu_{l}\left(  {dz}\right)  .
\end{align}
Solving the first order condition for $C$ leads to the optimal consumption
$C^{\ast}=x\left(  \gamma g(\boldsymbol{\lambda})\right)  ^{1/(\gamma-1)}$.

Finally, $g$ is characterized by the implicit equation
\begin{align}
\lbrack\mathcal{A}g](\boldsymbol{{\lambda}})-(\beta-r\gamma)g\left(
\boldsymbol{{\lambda}}\right)   &  =G(\boldsymbol{{\lambda}};g)\label{HJBg}\\
G({\boldsymbol{\lambda}};g)  &  =g\left(  {\boldsymbol{\lambda}}\right)
{K^{\gamma}}\left(  \boldsymbol{\omega}^{\ast}(\boldsymbol{h}%
(\boldsymbol{\lambda})),{\boldsymbol{h}}({\boldsymbol{\lambda}})\right)
+(1-\gamma){\left(  {\gamma g({\boldsymbol{\lambda}})}\right)  ^{\gamma
/(\gamma-1)}},\ \boldsymbol{{\lambda}}\in\mathbb{R}_{+}^{n}. \label{Gdef}%
\end{align}
The Markov generator $\mathcal{A}$ is again given by (\ref{Generator}). We use
the Feynman-Kac formula to write the solution $g$ as a fixed point of an
infinite dimensional nonlinear mapping:
\begin{align}
g  &  =\mathcal{G}(g)\label{fixedpoint}\\
(\mathcal{G}(g))\left(  \boldsymbol{\lambda}\right)  :=  &  \int_{0}^{\infty
}e^{-(\beta-r\gamma)s}\mathbb{E}_{0,\boldsymbol{\lambda}}\left[
G(\boldsymbol{\lambda}_{s};g)\right]  ds.
\end{align}

\begin{remark}
Note that the power case differs from the log case in two distinct ways. The
first term on the right side of (\ref{Gdef}) is a distorted version of the
right side of (\ref{HJBf}), where the $\lambda$ dependence in the function $K$
is distorted in a $g$ dependent fashion through the mapping $\boldsymbol{h}$.
The second term does not arise in the log case, and introduces complications;
a similar situation occurs and is dealt with in a different model, see
\cite{delongkluppelberg08}. Following \cite{delongkluppelberg08}, one can
attempt to verify that (\ref{fixedpoint}) is a contraction mapping, and that
consequently the sequence of iterates $\{g^{(i)},i=0,1,\dots\}$ with
$g^{(0)}=1$ and $g^{(i+1)}=\mathcal{G}(g^{(i)})$ converges to $g$. We do not
attempt this here.

In examples for which the functions $K$ and $\boldsymbol{\omega}_{t}^{\ast
}(\boldsymbol{\lambda})$ are explicitly solvable, the iteration scheme can
apparently be efficiently implemented numerically.
%Figure \ref{FIGURETOBEADDED} shows the resultant $g$ function in one such example.

\end{remark}

\subsection{Exponential Utility}

An investor with the exponential utility $U(x)=-e^{-\gamma x}/\gamma$ with
risk aversion parameter $\gamma>0$, unlike the log investor, can in principle
consume at a negative rate, perhaps even reaching negative wealth, and thus in
this setting the question of defining admissible strategies is more involved.
We can however, search for candidate optimal strategies that solve the reduced
HJB equation (\ref{eq:HJBgeneral}) in the form $L(x,\mathbb{\lambda}%
)=-\exp[-\kappa x]g(\mathbb{\lambda})$ for some positive function $g$, and
then attempt to interpret the result. The HJB equation turns into:%
\begin{align}
0  &  =\max_{\left\{  C,\boldsymbol{\omega}\right\}  }\left\{  U(C)-\beta
L+\frac{L}{g(\boldsymbol{\lambda})}[\mathcal{A}g](\boldsymbol{{\lambda}%
})\right. \nonumber\\
&  -\kappa L\left(  rx+\boldsymbol{\omega}^{\prime}\boldsymbol{R}x-C\right)
+\frac{\kappa^{2}L}{2}\boldsymbol{\omega}^{\prime}\boldsymbol{\Sigma\omega
}x^{2}\label{expHJB}\\
&  +\left.  L\sum_{l=1}^{m}\lambda_{l,t}\frac{g(\boldsymbol{\lambda
}+\boldsymbol{d}_{l})}{g(\boldsymbol{\lambda})}\int\left(  \exp\left[
-\kappa(\boldsymbol{\omega}^{\prime}\boldsymbol{J}_{l}zx\right]  -1\right)
\nu_{l}(dz)\right\}  \ ,\nonumber
\end{align}
where $\mathcal{A}g$ is given as before by (\ref{Generator}). The first order
conditions for $C^{\ast}$ imply that $\gamma U(C^{\ast})=\kappa L$ and hence
\[
C^{\ast}=\frac{1}{\gamma}\left(  \kappa x-\log g-\log\kappa\right)  .
\]
The candidate optimal portfolio weights are best expressed in terms of dollar
amounts $\pi_{i,t}=\omega_{i,t}X_{t}$ invested. We find $\boldsymbol{\pi}%
_{t}^{\ast}=\boldsymbol{\pi}^{\ast}(\boldsymbol{h}(\boldsymbol{\lambda
})),\boldsymbol{h}=(h_{1},\dots,h_{m})$ where
\begin{align}
\boldsymbol{\pi}^{\ast}(\boldsymbol{\lambda})  &  =\operatorname{argmin}%
_{\boldsymbol{\pi}}K(\boldsymbol{\pi},\boldsymbol{\lambda}))\\
{h_{l}}(\boldsymbol{{\lambda}})  &  ={\lambda_{l}}g(\boldsymbol{{\lambda}%
}+{\boldsymbol{{d}}_{l}})/g(\boldsymbol{{\lambda}})\\
K(\boldsymbol{\pi},\boldsymbol{\lambda})  &  =\kappa\boldsymbol{\pi}^{\prime
}\boldsymbol{R}-\frac{\kappa^{2}}{2}\boldsymbol{\pi}^{\prime}%
\boldsymbol{\Sigma\pi}-\sum_{l=1}^{m}\lambda_{l}\int\left(  \exp
[-\kappa\left(  \boldsymbol{\pi}^{\prime}\boldsymbol{J}\right)  _{l}%
z]-1\right)  \nu_{l}\left(  dz\right) \nonumber
\end{align}

Substitution of $C^{\ast},\boldsymbol{\pi}^{\ast}$ back into (\ref{expHJB})
leads to
\[
0=(r-\beta) g+\lbrack\mathcal{A}g](\boldsymbol{{\lambda}}) -\kappa\left(
rx-C^{\ast}\right)  -g(\boldsymbol{{\lambda}})K(\boldsymbol{\pi}^{\ast
}(\boldsymbol{\lambda}),\boldsymbol{h(\lambda}))
\]
which in turn implies that $\kappa=r\gamma$. The condition on $g$ is now
implicit:
\begin{align}
\lbrack\mathcal{A}g](\boldsymbol{{\lambda}}) -(\beta-r\gamma+r\log(r\gamma))
g(\boldsymbol{{\lambda}})  &  =\tilde G(\boldsymbol{{\lambda}};g)\\
\tilde G(\lambda;g)  &  = g(\boldsymbol{{\lambda}})\left[  r\log g+
K(\boldsymbol{\pi}^{\ast}(\boldsymbol{\lambda}),\boldsymbol{h(\lambda
}))\right]
\end{align}
This is very similar to the characterization of $g$ for the power utility
case. In examples where $K$ and $\boldsymbol{\pi}^{\ast}$ are explicitly
known, one can attempt to solve numerically for $g$ by iteration.

\section{Conclusions\label{sec:conclusions}}

This paper extends the range of models for which solutions to the optimal
dynamic portfolio-consumption problem are available, to one which includes
mutually exciting jumps. We analyze features of the optimal solution and show
that it differs from the usual case in important ways. In particular, it
introduces an explicit time-variation in the optimal portfolio weights in
response to changes in the jump intensity, providing a rare example of an
explicit time-varying optimal portfolio solution. Moreover, power and
exponential investors both adopt more aggressive strategies, characterized by
a distortion function $\boldsymbol{h}$, than the corresponding investor who
does not fully recognize the mutual excitation effect.\pagebreak

%elsevier}
\bibliographystyle{elsevier}
\bibliography{mainbib}

\end{document}